\RequirePackage{amsmath}
\newcommand{\shortversion}[1]{}
\newcommand{\longversion}[1]{#1}

\shortversion{%
\documentclass[11pt,envcountsame,a4paper]{llncs}
}
\longversion{%
\documentclass[10pt,a4paper]{article}%
\usepackage[hscale=0.7,scale=0.75]{geometry}%
}

\usepackage{dlback}

\longversion{%
\usepackage{url}
\newcommand{\email}[1]{\texttt{#1}}
\author{Johannes K. Fichte$^\dagger$ \and Arne Meier$^\ddagger$ \and Irina Schindler$^\ddagger$ \and\\
$^\dagger$ Technische Universität Wien, \email{jfichte@dbai.tuwien.ac.at}\\
$^\ddagger$ Leibniz Universität Hannover, \email{$\{$meier.schindler$\}$@thi.uni-hannover.de}
}
}
\shortversion{%
\author{Johannes K. Fichte\inst{1} \and Arne Meier\inst{2}\and Irina
  Schindler\inst{2}
}
\institute{Technische Universität Wien, \email{jfichte@dbai.tuwien.ac.at} \and Leibniz Universität Hannover, \email{$\{$meier.schindler$\}$@thi.uni-hannover.de}} 
}

\longversion{%
\usepackage{amsthm}
\newtheorem{lemma}{Lemma} 
\newtheorem{theorem}{Theorem} 
\newtheorem{corollary}{Corollary} 
 
\newtheorem{proposition}{Proposition} 
\newtheorem{definition}{Definition} 
\newtheorem{example}{Example} 
}
\title{Strong Backdoors for Default Logic}

\longversion{\longtrue}

\begin{document}
\maketitle

\begin{abstract}
 In this paper, we introduce a notion of backdoors to Reiter's propositional default logic and study structural properties of it. Also we consider the problems of backdoor detection (parameterised by the solution size) as well as backdoor evaluation (parameterised by the size of the given backdoor), for various kinds of target classes (\CNF, \horn, \krom, \monotone, \id). We show that backdoor detection is fixed-parameter tractable for the considered target classes, and backdoor evaluation is either fixed-parameter tractable, in $\para\Delta^P_2$, or in $\para\NP$, depending on the target class.
\end{abstract}

\section{Introduction}
In the area of non-monotonic logic one aims to find formalisms that
model human-sense reasoning. It turned out that this kind of reasoning is quite different
from classical deductive reasoning as in the classical approach the
addition of information always leads to an increase of derivable
knowledge. Yet, intuitively, human-sense reasoning does not work in
that way: the addition of further facts might violate previous
assumptions and can therefore significantly decrease the amount of
derivable conclusions. Hence, in contrast to the classical process the behaviour of human-sense reasoning 
is non-monotonic. In the 1980s, several kinds of formalisms have been
introduced, most notably, circumscription~\cite{mc80}, default logic~\cite{reiter:80}, autoepistemic logic~\cite{m85}, and non-monotonic logic \cite{md80}. A good introduction into this field is given by Marek and Truszczynśki~\cite{matr93}.

In this paper, we focus on Reiter's \emph{Default Logic (DL)}, which
has been introduced in 1980~\cite{reiter:80} and is one of the most
fundamental formalism for modelling human-sense reasoning. DL extends
the usual logical derivations by rules of default assumptions (\emph{default rules}). Informally, default rules follow the format ``in the absence of contrary information, assume~$\dots$''. Technically, these patterns
are taken up in triples of formulas $\frac{\alpha:\beta}{\gamma}$, which express ``if prerequisite $\alpha$ can be
deduced and justification $\beta$ is never violated then assume conclusion~$\gamma$''.
Default rules can be used to enrich calculi in different kinds of
logics. Here, we consider a variant of propositional formulas, namely,
formulas in conjunctive normal form~($\CNF$). A key concept of DL is that an application of default rules must not
lead to an inconsistency if conflicting rules are present, instead
such rules should be avoided if possible. This concept results in the notion
of \emph{stable extensions}, which can be seen as a maximally consistent view of an agent with respect to his knowledge base together in combination with its set of default rules.
The corresponding decision problem,~i.e., the extension existence problem, then asks whether a given default theory has a \emph{consistent stable extension}, and is the problem of our interest.
The computationally hard part of this problem lies
in the detection of the order and ``applicability'' of default rules, which is a quite
challenging task as witnessed by its $\SigmaP{2}$-completeness. In 1992, Gottlob
showed that many important decision problems, beyond the extension existence problem,  of non-monotonic logics are complete for the second level of the polynomial hierarchy~\cite{gottlob92} and thus are of high intractability.

A prominent approach to understand the intractability of a problem is to
use the framework of parameterised complexity, which was introduced by
Downey and Fellows~\cite{DowneyFellows13,DowneyFellows99}. The main
idea of parameterised complexity is to fix a certain structural
property (the \emph{parameter}) of a problem instance and to consider the
computational complexity of the problem in dependency of the
parameter. Then ideally, the complexity drops and the problem becomes
solvable in polynomial time when the parameter is fixed. Such problems
are called \emph{fixed-parameter tractable} and the corresponding
parameterised complexity class, which contains all fixed-parameter
tractable problems, is called~$\FPT$.
For instance, for the propositional satisfiability problem~(\SAT) one
(naïve) parameter is the number of variables of the given
formula. Then, for a given formula~$\varphi$ of size~$n$ and $k$
variables its satisfiability can be decided in time $O(n\cdot2^{k})$,~i.e., polynomial (even linear) runtime in $n$ if $k$ is considered to be fixed.

The invention of new parameters can be quite challenging, however, $\SAT$ has so far been considered under many different parameters~\cite{sz03,SamerSzeider08c,ccfxjkg04,ops13}.
A concept that provides a parameter and has been widely used in
theoretical investigations of propositional satisfiability are
backdoors~\cite{WilliamsGomesSelman03,gs12,GottlobSzeider08}.
The size of a
backdoor can be seen as a parameter with which one tries to exploit a
small distance of a formula from being tractable.
More detailed, given a class~$\form$ of formulas and a
formula~$\varphi$, a subset~$B$ of its variables is a \emph{strong
  $\form$-backdoor} if the formula~$\phi$ under every truth assignment
over $B$ yields a formula that belongs to the class~$\form$.
Using backdoors usually consists of two phases: (i)~finding a backdoor
(backdoor detection) and (ii)~using the backdoor to solve the problem
(backdoor evaluation).
If $\form$ is a class where $\SAT$ is tractable and backdoor detection
is fixed-parameter tractable for this class, like the class of all
Horn or Krom formulas, we can immediately conclude that $\SAT$ is
fixed-parameter tractable when parameterised by the size of a smallest
strong $\form$\hy backdoor.

\paragraph{Related Work.}%
Backdoors for propositional satisfiability have been introduced by
Williams, Gomes, and Selman~\cite{WilliamsGomesSelman03,WilliamsGomesSelman03a}. The
concept of backdoors has recently been lifted to some non-monotonic
formalisms as abduction~\cite{PfandlerRummeleSzeider13}, answer set
programming~\cite{fs15,fs15b}, and argumentation~\cite{DvorakOrdyniakSzeider12}. Beyond the classification of Gottlob \cite{gottlob92}, the complexity of fragments, in the sense of Post's lattice, has been considered by Beyersdorff et~al.\ extensively for default logic \cite{bmtv12}, and for autoepistemic logic by Creignou et~al.~\cite{cmtv12}. Also parameterised analyses of non-monotonic logics in the spirit of Courcelle's theorem~\cite{courcelle90OLD,courcelle90} have recently been considered by Meier et~al.~\cite{msstv15}. Further, Gottlob et~al.\ studied treewidth as a parameter for various non-monotonic logics~\cite{gopiwe10} and also considered a more CSP focused non-monotonic context within the parameterised complexity setting~\cite{gss02}.

\paragraph{Contribution.}%
In this paper, we introduce a notion of backdoors to propositional
default logic and study structural properties therein. 
Then we investigate the parameterised complexity of the problems of backdoor detection (parameterised by the solution size) and evaluation (parameterised by the size of the given backdoor), with respect to the most important classes of CNF
formulas,~e.g., $\CNF$, $\krom$, $\horn$, $\monotone$, and $\id$.
Informally, given a formula $\varphi$ and an integer $k$, the detection problem
asks whether there exists a backdoor of size $k$ for $\varphi$. Backdoor evaluation
then exploits the distance~$k$ for a target formula class to solve the
problem for the starting formula class with a ``simpler''
complexity. Our classification shows that detection is fixed-parameter
tractable for all considered target classes. However, for backdoor
evaluation starting at $\CNF$ the parameterised complexity depends, as expected, on the target class: the
parameterised complexity then varies between $\para\DeltaP2$ ($\monotone$), $\para\NP$ ($\krom, \horn$), and $\FPT$ ($\id$).

\section{Preliminaries}
We assume familiarity with standard notions in computational
complexity, the complexity classes~$\P$ and $\NP$ as well as the
polynomial hierarchy. For more detailed information, we refer to other standard 
sources~\cite{Papadimitriou94,flgr06,DowneyFellows13}.

\paragraph{Parameterised Complexity.} %
We follow the notion by Flum and Grohe \cite{fg03}. 
A \emph{parameterised (decision) problem}~$L$ is a subset of
$\Sigma^*\times\N$ for some finite alphabet~$\Sigma$.
%
%
%
  Let $C$ be a classical complexity class, then $\para C$ consists of
  all parameterised problems~$L\subseteq \Sigma^* \times \N$, for
  which there exists an alphabet~$\Sigma'$, a computable function
  $f\colon \N\to\Sigma'^*$, and a (classical) problem
  $L'\subseteq \Sigma^*\times\Sigma'^*$ such that (i)~$L'\in C$, and
  (ii)~for all instances $(x,k)\in\Sigma^*\times\N$ of $L$ we have
  $(x,k)\in L$ if and only if $(x,f(k))\in L'$.
%
%
  For the complexity class~$\P$, we write $\FPT$ instead of
  $\para\P$. We call a problem in $\FPT$ \emph{fixed-parameter
    tractable} and the runtime $f(k)\cdot|x|^{O(1)}$ also \emph{fpt-time}.
  Additionally, the parameterised counterparts of $\NP$ and
  $\DeltaP{2}=\P^\NP$, which are denoted by $\para\NP$ and
  $\para\DeltaP2$, are relevant in this paper.

\paragraph{Propositional Logic.} 
Next, we provide some notions from propositional logic. We consider a finite set of propositional variables 
and use the symbols~$\top$ and $\bot$ in the standard way. 
A \emph{literal} is a variable~$x$ (positive literal) or its
negation~$\neg x$ (negative literal). A \emph{clause} is a finite set
of literals, interpreted as the disjunction of these literals.  A
propositional formula in \emph{conjunctive normal form (CNF)} is a
finite set of clauses, interpreted as the conjunction of its clauses.
We denote the class of all CNF formulas by $\CNF$.  A clause is
\emph{Horn} if it contains at most one positive literal, \emph{Krom}
if it contains two literals, \emph{monotone} if it contains only
positive literals, and \emph{positive-unit} if it contains at most one
positive literal. We say that a CNF formula has a certain property if
all its clause have the property.  We consider several classes
of formulas in this paper. Table~\ref{tbl:forms} gives an overview on
these classes and defines clause forms for these classes.

\begin{table}[t]
	$$\shortversion{\begin{array}{l@{\hspace{0.5em}}p{4.8cm}l}\toprule}
	\longversion{\begin{array}{l@{\hspace{0.5em}}p{5cm}l}\toprule}
		\text{class} &clause description & \text{clause forms}\\\midrule
		$\CNF$ &no restrictions& \{\ell^+_1,\dots,\ell^+_n,\ell^-_1,\dots,\ell^-_m)\\
		$\horn$&at most one positive literal &\{\ell^+,\ell^-_1,\dots,\ell^-_n\},\{\ell^-_1,\dots,\ell^-_m\}\\
		$\krom$&binary clauses&\{\ell_1^+,\ell^+_2\},\{\ell^+,\ell^-\},\{\ell^-_1,\ell^-_2\}\\
		$\monotone$&no negation, just positive literals&\{\ell^+_1,\dots,\ell_n^+\}\\
		$\id$&only positive unit clauses&\{\ell^+\}\\
	\bottomrule
	\end{array}$$
	
	\caption{Considered normal forms. In the last row, $\ell^+_i$ denote positive, and $\ell^-_i$ negative literals; and $n$ and $m$ are integers such that $n,m\ge0$. }\label{tbl:forms}
\end{table}
%
%

A formula~$\varphi'$ is a \emph{subformula} of a $\CNF$ formula~$\varphi$ (in symbols $\varphi' \subseteq \varphi$) if for each clause~$C' \in \varphi'$ there is some clause~$C \in \varphi$ such that $C' \subseteq C$. We call a class~$\form$ of $\CNF$ formulas \emph{clause-induced} if whenever $F\in \form$, all subformulas $F'\subseteq F$ belong to $\form$. 
Note that all considered target classes in this paper are clause-induced.

Given a formula~$\varphi\in\CNF$, and a
subset~$X\subseteq\Vars\varphi$, then a \emph{(truth) assignment} is a
mapping~$\theta\colon X\to\{0,1\}$. The \emph{truth (evaluation)} of
propositional formulas is defined in the standard way, in particular,
$\theta(\bot) = 0$ and $\theta(\top)=1$.
We extend $\theta$ to literals by setting
$\theta(\neg x) = 1 - \theta(x)$ for $x \in X$.
By $\A(X)$ we denote the set of all
assignments~$\theta: X \to \{0,1\}$. For simplicity of presentation,
we sometimes identify the set of all assignments by its corresponding
literals,~i.e.,
$\A(X)=\set{\{\ell_1,\dots,\ell_{|X|}\}}{x\in X, \ell_i\in\{x,\lnot
  x\}}$. We write $\varphi[\theta]$ for the \emph{reduct} of $\phi$
where every literal~$\ell \in X$ is replaced by $\top$ if
$\theta(\ell)=1$, 
%
%
then all clauses that contain a literal~$\ell$ with $\theta(\ell) = 1$
are removed and from the remaining clauses all literals~$\ell'$ with
$\theta(\ell')=0$ are removed.  We say $\theta$ \emph{satisfies} $\phi$ if
$\varphi[\theta]\equiv\top$, $\phi$ is \emph{satisfiable} if there
exists an assignment that satisfies $\phi$, and $\phi$ is
\emph{tautological} if all assignments~$\theta \in \A(X)$ satisfy
$\phi$.
%
%
%
%
Let $\varphi,\psi\in\CNF$ and $X = \Vars{\phi} \cup \Vars{\psi}$. We
write $\varphi\models\psi$ if and only if for all assignments
$\theta \in \A(X)$ it holds that all assignments~$\theta$ that satisfy
$\phi$ also satisfy $\psi$.
%
Further, we define
$\Th{\varphi}\dfn\set{\psi\in\CNF}{\varphi\models\psi}$.
%
%

%

Note that any assignment $\theta\colon\Vars\varphi\to\{0,1\}$ can be
also represented by the CNF formula
$\bigwedge_{\theta(x)=1}x\land\bigwedge_{\theta(x)=0}\lnot x$.
Therefore, we often write $\theta\models\varphi$ if
$\varphi[\theta]\equiv\top$ holds.

We denote with $\SAT(\form)$ the problem, given a propositional formula $\phi\in\form$ asking whether $\phi$ is satisfiable. The problem $\TAUT(\form)$ is defined over a given formula $\phi\in\form$ asking whether $\phi$ tautological.


\subsection{Default Logic} 
%
We follow notions by Reiter \cite{reiter:80} and define a
\emph{default rule} $\delta$ as a
triple~$\frac{\alpha:\beta}{\gamma}$; $\alpha$ is called the
\textit{prerequisite}, $\beta$ is called the \textit{justification},
and $\gamma$ is called the \textit{conclusion}; we set
$\prereq(\delta)\dfn\alpha$, $\just(\delta)\dfn \beta$, and
$\concl(\delta)\dfn \gamma$.
%
%
%
%
%
If $\form$ is a class of formulas, then $\frac{\alpha:\beta}{\gamma}$
is an $\form$-default rule if $\alpha,\beta,\gamma\in\form$. An
$\form$-\textit{default theory} $\encoding{W,D}$ consists of a set of
propositional formulas $W\in \form$ and a set~$D$ of $\form$-default
rules. 
We sometimes call $W$ the
\emph{knowledge base} of $\encoding{W,D}$.  Whenever we do not
explicitly state the class~$\form$, we assume it to be $\CNF$.

\begin{definition}[Fixed point semantics, \cite{reiter:80}]
  Let $\encoding{W,D}$ be a default theory and $E$ be a set of
  formulas. Then $\Gamma(E)$ is the smallest set of formulas such that:
  \begin{enumerate}
  \item $W\subseteq\Gamma(E)$,
  \item $\Gamma(E)=\Th{\Gamma(E)}$, and 
  \item for each 
    $\frac{\alpha:\beta}{\gamma}\in D$ with $\alpha\in\Gamma(E)$ and
    $\lnot\beta\notin E$, 
    it holds that $\gamma\in\Gamma(E)$.
  \end{enumerate}
  $E$ is a \emph{stable extension} of $\encoding{W,D}$, if
  $E=\Gamma(E)$.
  %
  %
An extension is \emph{inconsistent} if it contains $\bot$, otherwise it is called \emph{consistent}. 
\end{definition}

A definition for stable extensions beyond fixed point semantics, which
has been introduced by Reiter~\cite{reiter:80} as well, uses the
principle of a stage construction.

\begin{proposition}[Stage construction, \cite{reiter:80}]\label{prop:stageConstruction}
  Let $\encoding{W,D}$ be a default theory and $E$ be a set of
  formulas. Then define $E_0\dfn W$ and 
  %
  %
  \[E_{i+1} \dfn \Th{E_i}\cup\left\{\,\gamma\;\Bigg|\;
      \frac{\alpha:\beta}{\gamma}\in D,\alpha\in E_i\text{ and }
      \lnot\beta\notin E\,\right\}.\]
  $E$ is a \emph{stable extension} of $\encoding{W,D}$ if and only if
  $E=\bigcup_{i\in\N}E_i$. The set
  %
  %
  %
  \[G=\left\{\,\frac{\alpha:\beta}{\gamma}\in D\;\Bigg|\; \alpha\in
      E\land\lnot\beta\notin E\,\right\}\]
  is called the set of \emph{generating defaults}. If $E$ is a stable
  extension of $\encoding{W,D}$, then
  $E=\Th{W\cup \set{\concl(\delta)}{\delta\in G}}$.
\end{proposition}

\begin{example}\label{ex:running}
  \sloppypar
  Let $W=\emptyset$, $W'=\{x\}$,
  $D_1=\{\frac{x:y}{\lnot y},\frac{\lnot x:y}{\lnot y}\}$, and
  $D_2=\{\frac{x:z}{\lnot y},\frac{x:y}{\lnot z}\}$.
  %
  %
  %
  %
  The default theory $\encoding{W,D_1}$ has only the stable
  extension
  $\Th{W}$.
  The default theory $\encoding{\{x\},D_1}$ has no stable extension.
  The default theory $\encoding{\{x\},D_2}$ has the stable
  extensions
  $\Th{\{x,\lnot y\}}$ and $\Th{\{x,\lnot z\}}$.
\end{example}

The following example illustrates that a default theory might contain ``contradicting'' 
default rules that cannot be avoided in the process of determining extension existence.  Informally, such default
rules prohibit stable extensions. Note that there are also less obvious situations where
``chains'' of such default rules interact with each other.

\begin{example}\label{ex:contradicting}
  Consider $W'$ and $D_2$ from Example~\ref{ex:running} and
  let~$D_2'=D_2 \cup \{\frac{\true:\beta}{\lnot\beta}\}$ for some formula $\beta$.  The default
  theory~$\encoding{W',D_2'}$ has no stable extension $\Th{W}$ unless $W\cup\{\lnot y\}\models\lnot\beta$ or $W\cup\{\lnot z\}\models\lnot\beta$.
\end{example}

Technically, the definition of stable extensions allows inconsistent stable extensions. However, Marek and Truszczyński have shown that inconsistent extensions only occur if the set~$W$ is already inconsistent where $\encoding{W,D}$ is the theory of interest \cite[Corollary 3.60]{matr93}. An immediate consequence of this result explains the interplay between consistency and stability of extensions more subtle: (i) If $W$ is consistent, then every stable extension of $\encoding{W,D}$ is consistent, and (ii) If $W$ is inconsistent, then $\encoding{W,D}$ has a stable extension. In Case~(2) the stable extension consists of all formulas~$\allPropForm$.
Hence, it makes sense to consider only consistent stable extensions as
the relevant ones. Moreover, we refer by $\SE{\encoding{W,D}}$ to the set
of all consistent stable extensions of~$\encoding{W,D}$.


A main computational problem for DL is the \emph{extension existence
  problem}, defined as follows where $\form$ is a class of
propositional formulas:
\begin{samepage}
  \decisionproblem {\Ext(\form)} {An $\form$\hy default theory
    $\langle W, D \rangle$.}  {Does $\langle W, D \rangle$ have a
    consistent stable extension?}
\end{samepage}

%

The following proposition summarises relevant results for the
extension existence problem for certain classes of formulas.

\begin{proposition}\label{prop:EXT-fragments}
  ~\\[-1.3em]
  \begin{enumerate}
  \item $\Ext(\CNF)$ is $\SigmaP{2}$-complete~\cite{gottlob92}.
  \item $\Ext(\horn)$ is
    $\NP$-complete~\cite{stillman,stillmanhorn-nud}.
  \item $\Ext(\id)\in\P$~\cite{bmtv12}.
  \end{enumerate}
\end{proposition}

\subsection{The Implication Problem}

The implication problem is an important (sub\hy)problem when reasoning
with default theories. In the following, we first formally introduce
the implication problem for classes of propositional formulas, and
then state its (classical) computational complexity for the
classes~$\horn$ and $\krom$.

\decisionproblem
{$\IMP(\form)$}
{A set~$\Phi$ of $\form$-formulas and a formula~$\psi\in\form$.}
{Does $\Phi \models \psi$ hold?}

Beyersdorff et~al.\ \cite{bmtv12} have considered all Boolean fragments of $\IMP(\form)$  and completely classified its computational complexity concerning the framework of Post's lattice. However, Post's lattice talks only about restrictions on allowed Boolean functions. Since several subclasses  of $\CNF$, like $\horn$ or $\krom$, use the Boolean functions ``$\land$'',''$\lnot$'', and ``$\lor$'', such classes are unrestricted from the perspective of Post's lattice. Still, efficient algorithms are known for such classes from propositional satisfiability. The next results state a similar behaviour for the implication problem. \iflong\else The proof can be found in the report \cite{report}.\fi

\begin{lemma}\label{lem:KromImpInP}
  $\IMP(\krom)\in\P$.
\end{lemma}\iflong
\begin{proof}
  Given a set $\Phi$ of $\krom$-formulas and a formula $\psi\in\krom$. Without loss of generality assume that $\bigwedge_{\varphi\in\Phi}\varphi = \bigwedge_{i=1}^m C_i$, and $\phi = \bigwedge_{i=1}^n C'_i$. Then it holds that
  \begin{align}
    \encoding{\Phi,\psi}\in\IMP(\krom) &\Leftrightarrow \left(\bigwedge_{i=1}^m C_i, \bigwedge_{i=1}^n C'_i\right)\in\IMP(\krom)\displaybreak[1]\\
		&\Leftrightarrow \left(\bigwedge_{i=1}^m C_i\right)\to \left(\bigwedge_{i=1}^n C'_i\right)\in\TAUT\displaybreak[1]\\
		&\Leftrightarrow \bigwedge_{i=1}^n \left(\bigwedge_{j=1}^m C_j\to C'_i\right)\in\TAUT\displaybreak[1]\\
		&\Leftrightarrow\forall\; 1\leq i\leq n \left(\bigwedge_{j=1}^m C_j\to C'_i\right)\in\TAUT\displaybreak[1]\\
		&\Leftrightarrow\lnot\exists\;1\leq i\leq n \left(\bigwedge_{j=1}^m C_j\to C'_i\right)\not\in\TAUT
  \end{align}	
(1) definition of the implication problem. (2) expressing implication through the propositional function $\to$. (3) $\alpha\to \beta\land\gamma$ is a tautology if and only if $(\alpha\to\beta)\land(\alpha\to\gamma)$ is a tautology. (4) separated to separate tautology questions. (5) $\alpha\land\beta$ is a tautology if neither $\alpha$ nor $\beta$ is not a tautology.

Now, we can check the last $n$ problems separately by
\begin{align*}
	\left(\bigwedge_{i=1}^m C_i\to (\ell\lor\ell')\right)\not\in\TAUT
	&\Leftrightarrow \left(\bigwedge_{i=1}^m C_i\right)[\theta_0]\in\SAT(\krom),
\end{align*}
where $\theta_0$ is the assignment such that $\theta_0(\ell)\dfn0$ and $\theta_0(\ell')\dfn0$. Observe that, if $\ell\equiv \sim\!\!\ell'$ then the implication on the left part of the equivalence is always a tautology.\shortversion{\qed}
\end{proof}\fi

Similar to the proof of Lemma~\ref{lem:KromImpInP} one can show the
same complexity for the implication problem of $\horn$
formulas. However, its complexity is already known from the work by
Stillman~\cite{stillman}.
\begin{proposition}[{\cite[Lemma 2.3]{stillman}}]\label{prop:imp-horn-p}
	$\IMP(\horn)\in\P$.
\end{proposition}

\section{Strong Backdoors} 

In this section, we lift the concept of backdoors to the world of
default logic. First, we review backdoors from the propositional
setting~\cite{WilliamsGomesSelman03,WilliamsGomesSelman03a}, where a
backdoor is a subset of the variables of a given formula.
Formally, for a class~$\form$ of formulas and a formula~$\phi$, a
\emph{strong $\form$\hy backdoor} is a set~$B$ of variables such that
for all assignments~$\theta\in\A(B)$, it holds that
$\phi[\theta]\in\form$.

Backdoors in propositional satisfiability follow the binary character of truth assignments. Each variable of a given formula is considered to be either true or false.  However, reasoning in default logic has a ternary character. When we consider consistent stable extensions of a given default theory then one of the following three cases holds for some formula~$\varphi$ with respect to an extension~$E$: (i)~$\varphi$ is contained in $E$, (ii)~the negation $\lnot\varphi$ is contained in $E$, or (iii)~neither $\varphi$ nor $\lnot\varphi$ is contained in $E$ (e.g., for the theory $\encoding{\{x\},D_2}$, from Example~\ref{ex:running}, neither $b$ nor $\lnot b$ is contained in any of the two stable extensions, where $b$ is a variable).
Since we need to weave this trichotomous point of view into a backdoor definition for default logic, the original definition of backdoors cannot immediately be transferred (from the SAT setting) to the scene of default logic. The first step is a notion of \emph{extended literals} and \emph{reducts}. The latter step can be seen as a generalisation of assignment functions to our setting.

\begin{definition}[Extended literals and reducts]
An \emph{extended literal} is a literal or a fresh variable~$x_\varepsilon$. For convenience, we further define $\sim\!\!\ell=x$ if $\ell=\lnot x$ and $\sim\!\!\ell=\lnot x$ if $\ell=x$. Given a formula $\phi$ and an extended literal $\ell$, then the \emph{reduct} $\reduct{\phi}{\ell}$ is obtained from $\phi$ such that
\begin{enumerate}
	\item \emph{if $\ell$ is a literal:} then all clauses that contain $\ell$ are deleted and all literals~$\sim\!\!\ell$ are deleted from all clauses,
	\item \emph{if $\ell$ is $x_\varepsilon$:} then all occurrences of literals $\lnot x, x$ are deleted from all clauses.
\end{enumerate}

Let $\encoding{W,D}$ be a default theory and $\ell$ an extended literal, then
\[\reduct{W,D}{\ell}:=\left(\reduct{W}{\ell},\left\{\frac{\reduct{\alpha}{\ell}:\reduct{\beta}{x}}{\reduct{\gamma}{\ell}\land y_i}\;\left|\;\delta_i=\frac{\alpha:\beta}{\gamma}\in D\right.\right\}\right),\] where $y_i$ is a fresh proposition, and $\reduct{W}{\ell}$ is $\bigcup_{\omega\in W}\reduct{\omega}{\ell}$.
\end{definition}
%
Later (in the proof of Lemma~\ref{lem:reducts}), we will see why we need the $y_i$s.

In the next step, we incorporate the notion of extended literals into sets of assignments. Therefore, we introduce  \emph{threefold assignment sets}. Let $X$ be a set of variables, then we define
\[
\tA(X):=\{\{a_1,\dots,a_{|X|}\}\mid x\in X \text{ and }a_i\in\{x, \lnot x,x_\varepsilon\}
\}.
\] 

Technically, $\A(X)\subsetneq\tA(X)$ holds. However, $\tA(X)$
additionally contains variables~$x_\varepsilon$ that will behave as ``don't
care'' variables encompassing the trichotomous reasoning
approach explained above.
For $Y\in\tA(X)$ the reduct~$\reduct{W,D}{Y}$ is the consecutive application of all $\reduct{\cdot}{y}$ for $y\in Y$ to $\encoding{W,D}$. Observe that the order in which we apply the reducts to $\encoding{W,D}$ is not important.

The following proposition states that implication of formulas is invariant under adding conjuncts of fresh variables to the premise.

\begin{proposition}\label{prop:fresh-y-is-good}
  Let $\varphi,\psi\in\CNF$ be two formulas and
  $y\notin\Vars\varphi\cup\Vars\psi$. Then $\varphi\models\psi$ if and
  only if $\varphi\land y\models\psi$.
\end{proposition}

Now we show that implication for $\CNF$ formulas that do not contain tautological clauses is invariant under the application of ``deletion reducts'' $\reduct{\cdot}{x_\varepsilon}$. \iflong\else The proof of the following results can be found in an extended version~\cite{report}.\fi

\begin{lemma}\label{lem:imp_invariant_deletion}
  Let $\psi,\varphi \in \CNF$ be two formulas that do not contain
  tautological clauses.  If $\psi\models\varphi$, then
  $\reduct{\psi}{x_\varepsilon}\models\reduct{\varphi}{x_\varepsilon}$
  for every variable $x\in\Vars{\varphi}\cup\Vars{\psi}$.
\end{lemma}
\iflong
\begin{proof}
	Assume for contradiction that $\reduct{\psi}{x_\varepsilon}\not\models\reduct{\varphi}{x_\varepsilon}$. Then there exists an assignment~$\theta\colon\allowbreak\Vars{\reduct{\psi}{x_\varepsilon}}\cup\allowbreak\Vars{\reduct{\varphi}{x_\varepsilon}}\to\{0,1\}$ such that $\theta\models\reduct{\psi}{x_\varepsilon}$ but $\theta\not\models\reduct{\varphi}{x_\varepsilon}$.
	As $\theta\models\reduct{\psi}{x_\varepsilon}$ every arbitrary extension of $\theta$ satisfies $\psi$, in particular also any extension on $\{x\}\cup\Vars{\reduct{\psi}{x_\varepsilon}}\cup\Vars{\reduct{\varphi}{x_\varepsilon}}$. Denote such an extension by $\theta_x$. Yet, by $\psi\models\varphi$ we get $\theta_x\models\varphi$. As this holds for any arbitrary such $\theta_x$ the satisfiability of $\varphi$ is independent of setting $x$ wherefore $\theta_x\models\reduct{\varphi}{x_\varepsilon}$ as well. (Note that here it is crucial that we require $\varphi$ contain no tautological clauses.) As $x\notin\Vars{\reduct{\varphi}{x_\varepsilon}}$ holds we get $\theta\models\reduct{\varphi}{x_\varepsilon}$ which is a contradiction. Thus $\reduct{\psi}{x_\varepsilon}\models\reduct{\varphi}{x_\varepsilon}$.\shortversion{\qed}
\end{proof}
\fi

The next lemma shows that implication for $\CNF$ formulas is invariant under the application of reducts over $\A$. 

\begin{lemma}\label{lem:imp_inv_usual_reducts}
	Let $\psi,\varphi$ be two $\CNF$ formulas, and $X\subseteq\Vars{\psi}\cup\Vars\varphi$. If $\psi\models\varphi$, then  $\reduct{\psi}{Y}\models\reduct{\varphi}{Y}$ holds for every set $Y\in\A(X)$.
\end{lemma}\iflong
\begin{proof}
Let $\psi,\varphi$, and $X$ be as in the formulation of the lemma and assume that $\psi\models\varphi$ holds. Now fix an arbitrary $Y\in \A(X)$ and consider every assignment $\tau_Y\colon\Vars{\reduct{\psi}{Y}}\cup\Vars{\reduct{\varphi}{Y}}\to\{0,1\}$. Note that $\tau_Y$ is defined on $(\Vars{\psi}\cup\Vars{\varphi})\setminus Y$. 
Define $\tau\upharpoonright Y$ as the assignment $\tau$ extended by setting $\tau(x)\dfn 1$ if $x\in Y$, and $\tau(x)\dfn 0$ if $\lnot x\in Y$. Thus $\tau\upharpoonright Y$ completely agrees with $\tau$ on the variables in $Y$.

Then $\tau\upharpoonright Y \models\lnot \psi\lor\varphi$ holds by assumption as $\psi\models\varphi$. Then by an easy induction we get $\tau\upharpoonright Y\models\varphi$ if and only if $\tau\models\reduct{\varphi}{Y}$, and $\tau\upharpoonright Y\models \psi$ if and only if $\tau\models\reduct{\psi}{Y}$. Thus we get 
$$
\tau\models \reduct{\psi}{Y} \Longleftrightarrow \tau\upharpoonright Y \models \psi \Longrightarrow \tau\upharpoonright Y\models \varphi\Longleftrightarrow\tau\models\reduct{\varphi}{Y}
$$
and the lemma follows.	\shortversion{\qed}
\end{proof}\fi

We denote by $\BD\IMP(\CNF\to\form)$ the parameterised version of the
problem~$\IMP(\CNF)$ where additionally a strong $\form$\hy backdoor
is given and the parameter is the size of the strong $\form$-backdoor.

\begin{corollary}\label{cor:imp-cnf-fpt}
 Given a class~$\form\in\{\id,\horn,\krom\}$ of CNF formulas. Then $\BD\IMP(\CNF\to\form) \in \FPT$.
\end{corollary}
\iflong
\begin{proof}
Let $W,\varphi,X$ be the given input instance. Then the following $\FPT$ algorithm decides the problem $\BD\IMP(\CNF\to\form)$. For every assignment $Y\in\A(X)$ check if $\reduct{W}{Y}\models\reduct{\varphi}{Y}$. For the corresponding classes $\form$ these implication problems are all decidable in polynomial time; for $\krom$ see Lemma~\ref{lem:KromImpInP}, for $\horn$ see Proposition~\ref{prop:imp-horn-p}, and $\id$ is a special case of $\horn$. The correctness follows from Lemma~\ref{lem:imp_inv_usual_reducts}. Hence the corollary applies.\shortversion{\qed}
\end{proof}\fi

A combination of Lemma~\ref{lem:imp_invariant_deletion} and Lemma~\ref{lem:imp_inv_usual_reducts} yields a generalisation for CNF formulas that do not contain tautological clauses. Note that the crucial difference is the use of $\tA$ instead of $\A$ in the claim of the result.

\begin{corollary}\label{cor:reducts-maintain-imp}
	Let $\psi,\varphi$ be two $\CNF$ formulas that do not contain tautological clauses, and $X\subseteq\Vars{E}\cup\Vars\varphi$ be a set of variables. If $\psi\models\varphi$ then for every set $Y\in\tA(X)$ it holds $\reduct{\psi}{Y}\models\reduct{\varphi}{Y}$.
\end{corollary}

The following lemma is an important cornerstone for the upcoming
section. It intuitively states that we do not loose any stable extensions under the application of reducts.
Before we can start with the lemma we need to introduce a
bit of notion. For a set~$D=\{\delta_1,\dots,\delta_n\}$ of default
rules and a set~$E$ of formulas we define
$ \conclusions{D}{E}\dfn\{\concl(\delta_i)\mid 1\leq i\leq n,
\delta_i\in D, E\models y_i\}, $ that is, the set of conclusions of
default rules~$\delta_i$ such that $y_i$ is implied by all
formulas in $E$. Further, for a set~$X$ of variables, we will extend the notion for
$\SE{\cdot}$ as follows:
\[
\SE{\encoding{W,D},X}\dfn \bigcup_{Y\in \tA(X)}\{\Th{W\cup \conclusions{D}{E}}\mid E\in\SE{\reduct{W,D}{Y}}\}.
\]

\begin{lemma}\label{lem:reducts}
	Let $\encoding{W,D}$ be a $\CNF$ default theory with formulas that do not contain tautological clauses, and $X$ be a set of variables from $\Vars{W,D}$. Then $\SE{\encoding{W,D}}\subseteq\SE{\encoding{W,D},X}.$
\end{lemma}
\begin{proof}
	Let $\encoding{W,D}$ be the given default theory, $X\subseteq\Vars{W,D}$, and $E\in\SE{\encoding{W,D}}$ be a  consistent stable extension of $\encoding{W,D}$.

	Now suppose for contradiction that $E\notin\SE{\encoding{W,D},X}$. Further, let $G$ be the set of generating defaults of $E$ by Proposition~\ref{prop:stageConstruction}, and w.l.o.g.~let $G\dfn\{\delta_1,\dots,\delta_k\}$ also denote the order in which these defaults are applied. Thus it holds that $E=\Th{W\cup\{\concl(\delta)\mid\delta\in G\}}$. 
	Hence, $W\models\prereq(\delta_1)$ holds and further fix a $Y\in \tA(X)$ which agrees with $E$ on the implied literals from $\Vars{W,D}$, i.e., $x\in Y$ if $E\models x$ for $x\in\Vars{W,D}$, $\lnot x\in Y$ if $\models\lnot x$, and $x_\varepsilon\in Y$ otherwise. Then, by Corollary~\ref{cor:reducts-maintain-imp} we know that also $\bigwedge_{\omega\in W}\reduct{\omega}{Y}\models\reduct{\prereq(\delta_1)}{Y}$ is true. Furthermore, we get that
	$$
	\bigwedge_{\omega\in W}\reduct{\omega}{Y}\land\bigwedge_{1\leq j\leq i}\reduct{\concl(\delta_j)}{Y}\models \reduct{\prereq(\delta_{i+1})}{Y}
	$$
	holds for $i<k$. Thus, by definition of $\reduct{W,D}{Y}$, the reducts of the knowledge base $W$ and the derived conclusions together trivially imply the $y_i$s, i.e., it holds that
	$$
	\bigwedge_{\omega\in W}\reduct{\omega}{Y}\land\bigwedge_{1\leq i\leq k}\reduct{\concl(\delta_i)}{Y}\models \bigwedge_{1\leq i\leq k}y_i.
	$$
	
	As neither $E\models\prereq(\delta)$ holds for some $\delta\in D\setminus G$, nor $E\cup\{\concl(\delta)\mid \delta\in G\}\models \delta'$ is true for some $\delta'\in D\setminus G$, $E$ is a consistent set, and $Y$ agrees with $E$ on the implied variables from $\Vars{W,D}$, we get that no further default rule $\delta$ is triggered by $\reduct{W}{Y}$ or $\reduct{W\cup\{\concl(\delta)\mid\delta\in D\setminus G\}}{Y}$. 
	
	Further, it holds that no justification is violated as $E\models\lnot\beta$ for some $\beta\in\bigcup_{\delta\in G}\just(\delta)$ would imply that $\reduct{E}{Y}\models\lnot\reduct{\beta}{Y}$ also holds by Corollary~\ref{cor:reducts-maintain-imp}. 
	Thus, eventually $E'=\Th{\reduct{W}{Y}\cup \{\reduct{\concl(\delta)}{Y}\mid \delta\in G\}}$ is a stable extension with respect to $\reduct{W,D}{Y}$. But, 
 the set of conclusions of $G$ coincides with $\conclusions{D}{E'}$ wherefore 
	\begin{align*}
		E&=\Th{W\cup \{\concl(\delta)\mid\delta\in G\}}\\
		&=\Th{W\cup \conclusions{D}{E'}}\in\SE{\encoding{W,D},X}
	\end{align*}
	holds, which contradicts our assumption. Thus, the lemma applies.\shortversion{\qed}
\end{proof}

We have seen that it is important to disallow tautological clauses. However, the detection of this kind of clauses is possible in polynomial time. Therefore, we assume in the following that a given theory contains no tautological clauses.
This is not a very weak restriction  as (i)~$\phi\land C\equiv\phi$ for any tautological clause~$C$, and (ii)~$C\equiv\top$ for any tautological clause~$C$.

The following example illustrates how reducts maintain existence of stable extensions.


\begin{example}\label{ex:reducts}
	The default theory $\encoding{W,D}=\{\{x\},\{\frac{x:y}{\lnot y\lor x}\}\}$ has the extension $E\dfn\Th{x,\lnot y\lor x}$ and yields the following cases for the backdoor $B=\{x\}$:
$\reduct{W,D}{x} =\langle\{\top\}, \{\frac{\top:z}{y_1}\}\rangle$, yielding $\SE{\reduct{W,D}{x}}=\{\Th{y_1}\},$
and, both, $\reduct{W,D}{\neg x}$ and $\reduct{W,D}{x_{\varepsilon}}$ yield an empty set of stable extensions. Thus, with $\conclusions{D}{\Th{y_1}}=\{\lnot y\lor x\}$ we get $\Th{\{\lnot y\lor x\}\cup\{x\}}$ which is equivalent to the extension $E$ of $\encoding{W,D}$.
\end{example}

Now, we are in the position to present a definition of strong
backdoors for default logic.

\begin{definition}[Strong Backdoors for Default Logic]
	Given a $\CNF$ default theory $\encoding{W,D}$, a set~$B\subseteq\Vars{W,D}$ of variables, and a class~$\form$ of formulas. We say that $B$ is a \emph{strong $\form$\hy backdoor} if for each $Y\in \tA(B)$ the reduct $\reduct{W,D}{Y}$ is a $\form$ default theory.
\end{definition}

\section{Backdoor Evaluation}
In this section, we  investigate the evaluation of strong backdoors for the extension existence problem in default logic with respect to different classes of CNF formulas. Formally, the problem of strong backdoor evaluation for extension existence is defined as follows.
\pproblem
{$\SBDexpl(\form\to\form')$}
{An $\form$\hy default theory~$\langle W,D \rangle$ and \iflong\else\\\fi
a strong $\form'$\hy backdoor~$B\subseteq \vars(W) \cup \vars(D)$.}
{The size of the backdoor~$B$.}%
{Does $\langle W,D \rangle$ have a stable extension? }

First, we study the complexity of the ``extension checking problem'', which is a main task we need to accomplish when using backdoors as our approach following Lemma~\ref{lem:reducts} yields only ``stable extension candidates''.
Formally, given a default theory~$\encoding{W,D}$ and a finite set~$\Phi$ of formulas, $\checkExt$ asks whether $\Th{\Phi}\in\SE{\encoding{W,D}}$ holds.

Rosati \cite{rosati} classified the extension checking problem as complete for the complexity class $\Theta^\P_2=\DeltaP2[\log]$, which allows only logarithmic many oracle questions to an $\NP$ oracle. \iflong For further information on the complexity class $\Theta^\P_2$ we refer the reader to the survey article of Eiter and Gottlob \cite{eg97}.\fi  We will later see that a simpler version suffices for our complexity analysis. Therefore, we state in Algorithm~\ref{alg:extchecking} an adaption of Rosatis algorithm \cite[Figure 1]{rosati} to our notation showing containment (only) in $\DeltaP2$. 

\begin{proposition}[{\cite[Figure 1, Theorem 4]{rosati}}]\label{prop:ext-checking-deltap2}
$\checkExt\in\DeltaP2$.
\end{proposition}

In a way, extension checking can be compared to model checking in logic. In default logic the complexity of the extension existence problem $\Ext$ is twofold: using the approach of Proposition~\ref{prop:stageConstruction} (i) one has to non-deterministically guess the set (and ordering) of the generating defaults, and (ii) one has to verify whether the generating defaults lead to an extension. For (ii), one needs to answer quadratic many implication questions. Hence, the problem is in $\NP^\NP$. Thus, a straightforward approach for $\checkExt$ omits the non-determinism in (i) and achieves the result in $\P^\NP$.

\begin{algorithm}[t]
  \LinesNumbered
  \SetAlgoNoEnd
  \SetAlgoNoLine
  \DontPrintSemicolon
  \KwIn{Set~$E$ of formulas and a default theory $\encoding{W,D}$}
  \KwOut{True iff $E$ is a stable extension of $\encoding{W,D}$}
 $D'\eqdef \emptyset$\;
 \ForAll(\texttt{// (1) Classify unviolated justifications.}){$\frac{\alpha:\beta}{\gamma} \in D$}{
 		\lIf{$E\not\models \lnot \beta$}{
 			$D'\eqdef D' \cup \{\frac{\alpha:}{\gamma}\}$
 		}
 }
 \tcp{(2) Compute extension candidate of justification-free theory.}
 $E'\eqdef W$ \;
 \While{$E'$ did change in the last iteration}{
	\ForAll{$\frac{\alpha:}{\gamma}\in D'$}{
		\lIf{$E'\models\alpha$}{
			$E'\eqdef E'\land\gamma$
		}
    }
 }
\tcp{(3) Does the candidate match the extension?}
\leIf{$E\models E'$ \textbf{\emph{and}} $E'\models E$}{\Return true}{\Return false}
\caption{Extension checking algorithm \cite[Theorem 4]{rosati}}\label{alg:extchecking}
\end{algorithm}

\begin{theorem}\label{thm:eval-cnf-to-horn}
  $\SBDexpl(\CNF\to\horn) \in \paraNP$.
\end{theorem}
\begin{proof}
Let $\encoding{W,D}$ be a given $\CNF$ default theory and $B\subseteq\Vars{W,D}$ be the given backdoor. In order to evaluate the backdoor we have to consider the $|\tA(X)|=3^{|B|}$ many different reducts to $\horn$ default theories. For each of them we have to non-deterministically guess a set of generating defaults $G$. Then, we use Algorithm~\ref{alg:extchecking} to verify whether $W\land\bigwedge_{g\in G}g$ is a stable extension (extensions can be represented by generating defaults; see Proposition~\ref{prop:stageConstruction}). $\IMP(\horn)\in\P$ by  Proposition~\ref{prop:imp-horn-p}. Hence, stable extension checking is in $\P$ for $\horn$ formulas. Then, after finding an extension~$E$ with respect to the reduct default theory $\reduct{W,D}{Y}$, we need to compute the corresponding extension $E'$ with respect to the original default theory. Here we just need to verify simple implication questions of the form $E\models y_i$ for $1\leq i\leq |D|$. Next, we need to verify whether $E'$ is a valid extension for $\encoding{W,D}$ using Algorithm~\ref{alg:extchecking}. Note that Corollary~\ref{cor:imp-cnf-fpt} shows that the implication problem of propositional formulas parameterised by the size of the backdoor is in $\FPT$, hence we can compute the implication questions inline. As the length of the used formulas is bounded by the input size and the relevant parameter is the same as for the input this runs in fpt-time.
\iflong

\else\fi
Together this yields a $\paraNP$ algorithm. \iflong Algorithm~\ref{alg:generic-eval-algo} depicts a generic algorithm in pseudocode.\fi\shortversion{\qed}
\end{proof}

\iflong\begin{algorithm}[t]
  \LinesNumbered
  \SetAlgoNoEnd
  \SetAlgoNoLine
  \DontPrintSemicolon
  \KwIn{$\form$-default theory $\encoding{W,D}$, backdoor $B\subseteq\Vars{W,D}$}
	\For{$Y\in\tA(X)$}{
		construct set of generating defaults $G$ for $\form'$ default theory $\reduct{W,D}{Y}$\;
		\If{$E\dfn\bigwedge_{w\in\reduct{W}{Y}}w\land\bigwedge_{\frac{\alpha:\beta}{\gamma}\in G}\gamma$ is extension for $\reduct{W,D}{Y}$}{
			$E'\dfn \bigwedge_{\omega\in W}\omega\land \bigwedge_{c\in\conclusions{D}{E'}}c$ \tcp{always in P by construction}
			\lIf{$E'$ is extension for $\encoding{W,D}$}{
				\Return true
			}
		}
	}
	\Return false
	\caption{Generic algorithm for $\SBDexpl(\form\to\form')$}\label{alg:generic-eval-algo}
\end{algorithm}\fi

\iflong
\begin{corollary}\label{cor:eval-cnf-to-krom}
$\SBDexpl(\CNF\to\krom)\in\paraNP$.
\end{corollary}
\begin{proof}
The implication problem of $\krom$ formulas is in $\P$ due to Lemma~\ref{lem:KromImpInP}. 
Thus under a similar argumentation as in the proof of Theorem~\ref{thm:eval-cnf-to-horn} we can construct a $\paraNP$ algorithm.\shortversion{\qed}
\end{proof}

\begin{corollary}\label{cor:eval-cnf-to-mon}
$\SBDexpl(\CNF\to\monotone) \in \para\DeltaP2$.
\end{corollary}
\begin{proof}
	For a monotone formula $\varphi$ its negation is not any longer monotone unless $\varphi\in\{\top,\bot\}$. This observation is important for such $\varphi$ occurring as justifications. If $\varphi\notin\{\top,\bot\}$ then this justification can be deleted as its negation will not be inferable whence the default rule is applicable whenever its prerequisite is met. If $\varphi\in\{\top,\bot\}$ then either it is only applicable in an inconsistent case or always. Hence we can distinct between these cases in polynomial time. Further observe that because of the previous argumentation there exists a unique stable extension if any. Thus the construction of the set of generating defaults and also the extension is achievable in $\para\DeltaP2$ as we have to do quadratic many implication questions, and the implication problem for $\monotone$ formulas has the same upper bound as the unrestricted one, hence $\co\NP$. Step (5) of Algorithm~\ref{alg:generic-eval-algo} is then just uses Algorithm~\ref{alg:generic-eval-algo} for implication questions which are solved via the standard algorithm (which is possible as we use a $\para\DeltaP2$ algorithm). \shortversion{\qed}
\end{proof}

The following corollary shows that the consideration of backdoor evaluation for the extension existence problem starting from $\krom$ default theories is interesting.
\begin{corollary}\label{cor:ext-krom-np}
	$\EXT(\krom)$ is $\NP$-complete.
\end{corollary}
\begin{proof}
	As Lemma~\ref{lem:KromImpInP} shows that $\IMP(\krom)\in\P$ we get that the extension checking problem for $\krom$ default theories is in $\P$ with the help of Proposition~\ref{prop:ext-checking-deltap2}. In order to show the $\NP$ upper bound, on input~$\encoding{W,D}$ the algorithm just guesses the set of generating defaults $G\subseteq D$ and then verifies if $W\land\bigwedge_{\frac{\alpha:\beta}{\gamma}\in G}\gamma$ is an extension with respect to $\encoding{W,D}$.
	
	For the lower bound observe that the default theory constructed by Beyersdorff et~al.\ \cite[Lemma 5.6]{bmtv12} consists only of $\krom$ formulas settling the lower bound by an reduction from $\threeSAT$. \shortversion{\qed}
\end{proof}

\begin{corollary}
	$\SBDexpl(\CNF\to\id)\in \FPT$.
\end{corollary}

\begin{proof}
The implication problem for $\id$ formulas is in $\AC0$ by Beyersdorff et~al.\ who showed this result for formulas using only conjunctions \cite[Theorem 4.1(4)]{bmtv09}. Hence, Algorithm~\ref{alg:extchecking} runs in polynomial time. Thus we achieve the $\FPT$ upper bound by a similar argumentation is the proof of Theorem~\ref{thm:eval-cnf-to-horn}.\shortversion{\qed}
\end{proof}
\else
The following corollary summarises the results for the remaining classes. Detailed proofs can be found in an extended version~\cite{report}.
\begin{corollary}\label{cor:eval-rest}
  ~\\[-1.4em]
  \begin{enumerate}
  \item $\SBDexpl(\CNF\to\monotone) \in \para\DeltaP{2}$,
  \item $\SBDexpl(\CNF\to\krom) \in \para\NP$, and
  \item $\SBDexpl(\CNF\to\id) \in \FPT$.
  \end{enumerate}
\end{corollary}
\fi

\section{Backdoor Detection}\label{sec:detection}
In this section, we study the problem of finding backdoors, formalised
in terms of the following parameterised problem:

\pproblem%
{\BDdetect{\CNF}{\form}}%
{A $\CNF$ default theory~$T$ and an integer~$k$.}%
{The integer~$k$.}%
{Does $T$ have a strong $\form$-backdoor of size at most~$k$?}

If the target class~$\form'$ is clause-induced, we can use a decision
algorithm for~\BDdetect{\form}{\form'} to find the backdoor using
self-reduction~\cite{Schnorr81,DowneyFellows13}. \iflong\else The proof can be found in an extended version~\cite{report}.\fi

\begin{lemma}
  Let $\form$ be a clause-induced class of \CNF formulas.
  If \BDdetect{\CNF}{\form} is fixed-parameter tractable, then also
  computing a \strongBd{\form} of size at most~$k$ of a given default
  theory~$T$ is fixed-parameter tractable (for parameter~$k$).
\end{lemma}\iflong
\begin{proof}
  Let $T = \encoding{W,D}$ be a default theory. We proceed by
  induction on $k$. If $k=0$ the statement is clearly true. Let
  $k>0$. Given $(T,k)$ we check for all~$v \in \Vars W \cup \Vars D$
  whether $\reduct{W,D}{Y}$, $\reduct{W,D}{Y'}$, and
  $\reduct{W,D}{Y''}$ have a \strongBd{\form} of size at most~$k-1$
  where $Y=\{v\}$, $Y'=\{\neg v\}$, and $Y''=\{v_\epsilon\}$.  If the
  answer is No for all $v$, then $T$ has no \strongBd{\form} of
  size~$k$. If the answer is Yes for~$v$, then by induction hypothesis
  we can compute a \strongBd{\form}~$B$ of size at most~$k-1$ of
  $\reduct{W,D}{Y}$, $\reduct{W,D}{Y'}$, and $\reduct{W,D}{Y''}$ and
  $B \cup \{v\}$ is a \strongBd{\form} of $T$.\shortversion{\qed}
\end{proof}\fi

The following theorem provides interesting target classes, where we
can determining backdoors in fpt-time.

\newcommand{\CCC}{\ensuremath{\mathcal{C}}\xspace}

\begin{theorem}\label{thm:bd-detect}
  Let $\CCC \in \{\horn,$ \id, \krom, \monotone $\}$, then \iflong\else the problem\fi
  $\BDdetect{\CNF}{\CCC}\in \FPT$.
\end{theorem}
\begin{proof}
  Let $\encoding{W,D}$ be a \CNF default theory and
  $\mathsf{F}:=W \cup \SSB \prereq(\delta),\allowbreak
  \just(\delta),\allowbreak \concl(\delta) \SM \delta \in D \SSE$.
  Since each class~$\CCC \in \{\horn,$ \id, \krom, \monotone $\}$ is
  clause-induced and then obviously
  $\reduct{\phi}{Z} \subseteq \reduct{\phi}{Y}$ holds for any
  $Z \in \tA(X)$, we have to consider only the case
  $Y = \SSB x_\epsilon \SM x \in X \SSE$ to construct a
  $\strongBd{\CCC}$ of $\encoding{W,D}$.  Thus let
  $Y = \SSB x_\epsilon \SM x \in X \SSE$\iflong in the following.\else.\fi.

  $\CCC = \monotone$: %
  A \CNF formula~$\phi$ is monotone if every literal appears only
  positively in any clause~$C \in \phi$ where $\phi \in
  \mathsf{F}$. 
  We can trivially construct a smallest \strongBd{\monotone} by taking
  all negative literals of clauses in formulas of $\mathsf{F}$ in
  linear time. Hence, the claim holds.

  For $\CCC \in \{\horn, \id, \krom \}$ we follow known constructions
  from the propositional setting~\cite{SamerSzeider08c}. Therefore, we
  consider certain (hyper\hy)graph representations of the given theory
  and establish that a set~$B\subseteq \Vars{\mathsf{F}}$ is a \strongBd{\CCC}
  of $\encoding{W,D}$ if and only if $B$ is a $d$\hy hitting set of
  the respective (hyper\hy)graph representation of $\encoding{W,D}$
  where $d$ depends on the class of formulas,~i.e., $d=2$ for $\horn$
  and $\id$ and $d=3$ for $\krom$.
  %
  A 2\hy hitting set (\emph{vertex cover}) of a graph~$G=(V,E)$ is a
  set~$S \subseteq V$ such that for every edge~$uv\in E$ we have
  $\{u,v\} \cap S \neq \emptyset$.
  A \emph{3-hitting set} of a hypergraph~$H=(\mathsf{V},\mathsf{E})$,
  with $E \in \mathsf{E}$ and $\Card{E} \leq 3$, is a
  set~$S \subseteq \mathsf{V}$ such that for every
  hyperedge~$E \in \mathsf{E}$ we have $E \cap S \neq \emptyset$.
  Then, a vertex cover of size at most~$k$, if it exists, can be found
  in time~$O(1.2738^k + kn)$~\cite{ChenKanjXia10} and a 3-hitting set
  of size at most~$k$, if it exists, can be found in
  time~$O(2.179^k + n^3)$~\cite{Fernau10a}, which gives us then a
  \strongBd{\CCC} of~$\encoding{W,D}$. It remains to define the
  specific graph representations and to establish the connection to
  \strongBd{\CCC}s.

  Definition of the various (hyper\hy)graphs: For $\CCC = \horn$ we
  define a graph~$G^+_T$ on the set of variables of $\mathsf{F}$,
  where two distinct variables~$x$ and $y$ are joined by an edge if
  there is a formula~$\phi \in \mathsf{F}$ and some
  clause~$C \in \phi$ with $x,y \in C$.
  %
  %
  %
  For $\CCC = \id$ we define a graph~$G_T$ on the set of variables of
  $\mathsf{F}$, where two distinct variables~$x$ and $y$ are joined by
  an edge if there is a formula~$\phi \in \mathsf{F}$ and some
  clause~$C \in \phi$ with $l_x,l_y \in C$ where
  $l_x \in \{x, \neg x\}$ and $l_y \in \{y, \neg y\}$. 
  %
  For $\CCC = \krom$ we define a hypergraph~$H_T$ on the
  variables~$\Vars{\mathsf{F}}$ where distinct variables~$x$, $y$, $z$
  are joined by a hyperedge if there is a
  formula~$\phi \in \mathsf{F}$ and some clause $C \in \phi$ with
  $\{x,y,z\} \subseteq \Vars C$.
  %

\iflong
  Next, we establish the only-if direction of the claim:
  %
  Let $B\subseteq \Vars{\mathsf{F}}$ be a \strongBd{\CCC} of
  $\encoding{W,D}$.  Consider an edge~$uv$ of $G$. By construction of
  $G^+_T$, $G_T$, and $H_T$ there is a corresponding
  clause~$C \in \phi$ for some formula~$\phi \in \mathsf{F}$ with
  $u,v \in C$.  By assumption, we construct $\reduct{\phi}{Y}$ from
  $\phi$ by deleting all occurrences of literals~$\lnot x$ and $x$
  from clauses in~$\phi$.  Since each clause in $\reduct{\phi}{Y}$
  contains at most one positive literal (\horn), or only positive unit
  clauses (\id), or at most one variable (\krom), respectively, we
  have $\{u,v\}\cap X \neq \emptyset$.
  We conclude that $B$ is a vertex cover of $G^+_T$, vertex cover of
  $G_T$, or 3\hy hitting set of $H_T$, respectively, which establishes
  the only-if direction of the claim.

  Finally, we establish the if direction of the claim: Therefore,
  assume that $B$ is a $d$\hy hitting set of the graph respective
  (hyper\hy)graph representation ($d=2$ for $\horn$ and $\id$ and
  $d=3$ for $\krom$).
  %
  Consider a clause~$C \in \reduct{\phi}{Y}$
  for some~$\phi \in \mathsf{F}$. For proof by contradiction assume
  that $C$ is not Horn, or not positive unit, or not Krom,
  respectively.
  Then there is a set~$S\subseteq C$ ($\Card{V}=2$ for $\horn$ and
  $\id$ and $\Card{V}=3$ for $\krom$) and an edge~$S$ of $G$ such that
  $S \cap X =\emptyset$, contradicting the assumption that $B$ is a
  vertex cover or 3\hy hitting set, respectively. Hence the if
  direction of the claim holds, which establishes the theorem. \shortversion{\qed}
\else
Then it is straightforward to show the correctness of the reduction.\shortversion{\qed}
\fi
\end{proof}

Now, we can use Theorem~\ref{thm:bd-detect} to strengthen the results
of Theorem~\ref{thm:eval-cnf-to-horn} and
\iflong%
Corollaries~\ref{cor:eval-cnf-to-krom} and \ref{cor:eval-cnf-to-mon} %
\else
Corollary~\ref{cor:eval-rest}
\fi
by dropping the assumption that the backdoor is given.

\begin{corollary}
  Let $\CCC \in \{\horn,$ \krom, \monotone $\}$, then the problem
  $\SBDexpl(\CNF\to\CCC)$ is in $\paraNP$ when parameterised by the
  size of a smallest \strongBd{\CCC} of the given theory. Further, the
  problem $\SBDexpl(\CNF\to\id)$ is in $\FPT$ when parameterised by
  the size of a smallest \strongBd{\id} of the given theory.
\end{corollary}

\section{Conclusion}
We have introduced a notion of strong backdoors for propositional
default logic. In particular, we investigated on the parameterised
decision problems backdoor detection and backdoor evaluation. We have
established that backdoor detection for the classes~$\CNF$, $\horn$,
$\krom$, $\monotone$, and $\id$ are fixed-parameter tractable whereas
for evaluation the classification is more complex. If $\CNF$ is the
starting class and $\horn$ or $\krom$ is the target class, then
backdoor evaluation is in $\para\NP$. If $\monotone$ is the target
class, then backdoor evaluation is in $\para\DeltaP2$, which is can be
solved by an fpt-algorithm that can query a SAT solver multiple
times~\cite{DeHaanSzeider14a}. For $\id$ as target class backdoor
evaluation is fixed-parameter tractable. 

An interesting task for future research is to consider the remaining
Schaefer classes~\cite{schaefer78},~e.g., dual-Horn, 1- and 0-valid,
as well as the classes renamable-Horn and
QHorn~\cite{bocrha90,Boros19941}, and investigate whether we can
generalise 
\iflong%
Algorithm~\ref{alg:generic-eval-algo}.%
\else%
the concept of Theorem~\ref{thm:eval-cnf-to-horn}.
\fi 
We have established for backdoor evaluation the upper bounds $\para\NP$ and
$\para\DeltaP2$, respectively. We think that it would also be
interesting to establish corresponding lower bounds. 
Finally, a direct application of quantified Boolean formulas in the
context of propositional default logic, for instance, via the work of
Egly et~al.~\cite{ettw00} or exploiting backdoors similar to results
by Fichte and Szeider~\cite{fs15b}, might yield new insights.

\subsection*{Acknowledgements}
The first author gratefully acknowledges support by the Austrian
Science Fund (FWF), Grant Y698. He is also affiliated with the
Institute of Computer Science and Computational Science at University
of Potsdam, Germany. The second and third author gratefully
acknowledge support by the German Research Foundation (DFG), Grant
ME 4279/1-1. The authors thank Jonni Virtema for pointing out
Lemma~\ref{lem:KromImpInP} and Sebastian Ordyniak for discussions on
Lemma~\ref{lem:imp_invariant_deletion}.

\bibliographystyle{plain}
\bibliography{dlback}

\addcontentsline{toc}{chapter}{Bibliography}

\end{document}